\documentclass[a4paper,UKenglish]{lipics-v2018}

\usepackage[utf8]{inputenc}
\usepackage{amsmath}
\usepackage{amssymb}
\usepackage{amsfonts}
\usepackage{latexsym}
\usepackage{xspace}
\usepackage[]{todonotes}
\usepackage{color}
\usepackage{hyperref}
\usepackage{graphicx}
\usepackage{csquotes}

\usepackage{wrapfig}

\title{Stabbing Rectangles by Line Segments~-- 
  How~Decomposition Reduces the Shallow-Cell Complexity}

\titlerunning{Stabbing Rectangles by Line Segments}

\author{Timothy M.\ Chan}{University of Illinois at Urbana-Champaign,
  U.S.A.}{tmc@illinois.edu}{}{}

\author{Thomas C.\ van Dijk}{Universit{\"a}t W{\"u}rzburg,
  W{\"u}rzburg, Germany}{}{https://orcid.org/0000-0001-6553-7317}{}

\author{Krzysztof Fleszar}{Max-Planck-Institut f{\"u}r Informatik,
  Saarbr{\"u}cken, Germany}{}{https://orcid.org/0000-0002-1129-3289}{This research was partially supported by Conicyt Grant PII 20150140 and by Millennium Nucleus Information and Coordination in Networks RC130003.}

\author{Joachim~Spoerhase}{Universit\"at W\"urzburg, W\"urzburg,
  Germany}{joachim.spoerhase@uni-wuerzburg.de}{https://orcid.org/0000-0002-2601-6452}{}

\author{Alexander~Wolff}{Universit\"at W\"urzburg, W\"urzburg,
  Germany}{}{https://orcid.org/0000-0001-5872-718X}{}

\authorrunning{T.\,M.~Chan, T.\,C.~van Dijk, K.~Fleszar, J.~Spoerhase,
  and A.~Wolff} 

\Copyright{Timothy M.~Chan, Thomas C.~van Dijk, Krzysztof~Fleszar,
  Joachim~Spoerhase, Alexander~Wolff}

\subjclass{F.2.2 Nonnumerical Algorithms and Problems}

\keywords{Geometric optimization, NP-hard, approximation, shallow cell
  complexity, line stabbing}

\hideLIPIcs

\newtheorem{observation}{\bfseries Observation}

\newcommand{\emphasizedAlreadyInAbstract}[1]{#1}
\newcommand{\runtimeclass}[1]{\ensuremath{\mathsf{#1}}}
\newcommand{\FPT}{\runtimeclass{FPT}\xspace}
\newcommand{\NP}{\runtimeclass{NP}\xspace}
\newcommand{\APX}{\runtimeclass{APX}\xspace}
\newcommand{\manuscript}{paper\xspace}

\newcommand{\Fig}{Figure\xspace}
\newcommand{\fig}{Fig.\xspace}
\newcommand{\figs}{Figs.\xspace}
\newcommand{\bigOh}{O}

\newcommand{\finO}{=}

\newcommand{\thSuffix}{-th\xspace}
\newcommand{\len}[1]{\left\|#1\right\|}

\newcommand{\ceil}[1]{\left\lceil #1 \right\rceil}

\newcommand{\formulaPunctuationSpace}{~}
\newcommand{\problemName}[1]{\textsc{#1}}
\newcommand{\GMMN}{\problemName{Generalized Minimum Manhattan Network}\xspace}

\newcommand{\Stabbing}{\problemName{Stabbing}\xspace}
\newcommand{\cost}{\mathrm{cost}}

\newcommand{\ConstrainedStabbing}{\problemName{Constrained Stabbing}\xspace}
\newcommand{\CardinalityStabbing}{\problemName{Cardinality Stabbing}\xspace}

\renewcommand{\paragraph}[1]{\medskip\noindent{\sffamily\bfseries #1}~}
\renewcommand{\phi}{\varphi}
\renewcommand{\epsilon}{\varepsilon}
\renewcommand{\rho}{\varrho}

\newcommand{\opt}{\mathrm{OPT}}

\newcommand{\fixC}{c}
\newcommand{\RvG}[1]{R_{#1}} 
\newcommand{\ReG}[1]{r_{#1}} 
\newcommand{\Sopt}{S_\opt}
\newcommand{\rTop}{r_\mathrm{top}}
\newcommand{\rBot}{r_\mathrm{bot}}
\newcommand{\Sact}[1]{S_\mathrm{act}^{#1}}
\newcommand{\Sina}[1]{S_\mathrm{ina}^{#1}}
\newcommand{\lenAct}[1]{\len{{\Sact{#1}}}}
\newcommand{\lenIna}[1]{\len{{\Sina{#1}}}}
\newcommand{\w}{\operatorname{width}}
\newcommand{\SvOpt}[1]{\Sopt^{#1}}

\newcommand{\vOpt}[1]{\len{{\SvOpt{#1}}}}

\newcommand{\SPSC}{\problemName{Special-3sc}\xspace}
\newcommand{\intAs}{\mathcal{I}} 
\newcommand{\shortSSC}{\textrm{3SC}}
\newcommand{\shortStab}{\textrm{stab}}
\newcommand{\xSSC}{I_\shortSSC}
\newcommand{\xStab}{I_\shortStab}
\newcommand{\copt}{\mathrm{OPT}}
\newcommand{\coptSSC}{\copt}
\newcommand{\coptStab}{\copt}
\newcommand{\ySSC}{S_\shortSSC}
\newcommand{\yStab}{S_\shortStab}
\newcommand{\costSSC}{\cost}
\newcommand{\costStab}{\cost}

\newcommand{\sets}{\ensuremath{\mathcal{F}}}
\newcommand{\universe}{\ensuremath{U}}
\newcommand{\unipoint}{\ensuremath{p_\mathrm{u}}}
\newcommand{\theTick}{\includegraphics{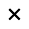}}
\newcommand{\Ll}{\ensuremath{x_\mathrm{1}}}
\newcommand{\Lr}{\ensuremath{x_\mathrm{2}}}
\newcommand{\Ly}{\ensuremath{y}}
\newcommand{\Rl}{\ensuremath{x_\mathrm{min}}}
\newcommand{\Rr}{\ensuremath{x_\mathrm{max}}}
\newcommand{\Rb}{\ensuremath{y_\mathrm{min}}}
\newcommand{\Rt}{\ensuremath{y_\mathrm{max}}}

\usepackage[inline]{enumitem}
\newenvironment{listRoman}{\begin{enumerate}[topsep=.5ex,label=(\roman*),align=parleft,leftmargin=*]}{\end{enumerate}}
\newenvironment{listArabic}{\begin{enumerate}[label=(\arabic*)]}{\end{enumerate}}

\graphicspath{{pictures/}}

\begin{document}
\maketitle

\begin{abstract}
  We initiate the study of the following natural geometric
  optimization problem.  The input is a set of axis-aligned rectangles
  in the plane.  The objective is to find a set of horizontal line
  segments of minimum total length so that every rectangle is
  \emph{stabbed} by some line segment.  A line segment stabs a
  rectangle if it intersects its left and its right boundary.  The
  problem, which we call \Stabbing, can be motivated by a resource
  allocation problem and has applications in geometric network design.
  To the best of our knowledge, only special cases of this
  problem have been considered so far.

  \Stabbing is a weighted geometric set cover problem, which we show
  to be \NP-hard.  A constrained variant of \Stabbing turns out to be
  even \APX-hard. While for general set cover the best possible
  approximation ratio is $\Theta(\log n)$, it is an important field in
  geometric approximation algorithms to obtain better ratios for
  geometric set cover problems. Chan et al.\ [SODA'12] generalize
  earlier results by Varadarajan [STOC'10] to obtain sub-logarithmic
  performances for a broad class of \emph{weighted} geometric set
  cover instances that are characterized by having low
  \emph{shallow-cell complexity}.  The shallow-cell complexity of
  \Stabbing instances, however, can be high so that a direct
  application of the framework of Chan et al.\ gives only logarithmic
  bounds.  We still achieve a constant-factor approximation by
  decomposing general instances into what we call \emph{laminar}
  instances that have low enough complexity.  

  Our decomposition technique yields constant-factor approximations
  also for the variant where rectangles can be stabbed by horizontal
  and vertical segments and for two further geometric set cover
  problems.
\end{abstract}

\section{Introduction}

In this \manuscript, we study the following geometric
optimization problem, which we call \Stabbing.  The input is a
set~${R}$ of~${n}$ axis-aligned rectangles in the plane.  The
objective is to find a set~${S}$ of horizontal line segments of
minimum total length~${\len{S}}$, where~${\len{S}=\sum_{s\in
    S}\len{s}}$, such that each rectangle~${r\in R}$ is
\emphasizedAlreadyInAbstract{stabbed} by some line segment~${s\in
  S}$. Here, we say that~${s}$ stabs~${r}$ if~${s}$
intersects the left and the right edge of~${r}$ (see
\fig~\ref{fig:stabbing-example}).  The length of a line segment~${s}$
is denoted by~${\len{s}}$.  Throughout this \manuscript, rectangles
are assumed to be axis-aligned and segments are horizontal line
segments (unless explicitly stated otherwise).

\begin{wrapfigure}{R}{5cm}
	\centering
	\includegraphics{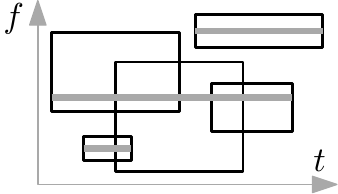}
	\caption{An instance of \Stabbing (rectangles) with an optimal
		solution (gray line segments).}
	\label{fig:stabbing-example}
\end{wrapfigure}

Our problem can be viewed as a resource allocation problem.  Consider
a server that receives a number of communication requests.  Each
request~${r}$ is specified by a time window~${[t_1,t_2]}$ and a
frequency band~${[f_1,f_2]}$.  In order to satisfy the request~${r}$,
the server has to open a communication channel that is available in
the time interval~${[t_1,t_2]}$ and operates at a fixed frequency
within the frequency band~${[f_1,f_2]}$.  Therefore, the server has to
open several channels over time so that each request can be
fulfilled. Requests may share the same channel if their frequency
bands and time windows overlap.  Each open channel incurs a fixed cost
per time unit and the goal is to minimize the total cost.  Consider
a~${t}$--$f$ coordinate system.  A request~${r}$ can be identified
with a rectangle~${[t_1,t_2]\times[f_1,f_2]}$.  An open channel
corresponds to horizontal line segments and the operation cost equals
its length.  Satisfying a request is equivalent to stabbing the
corresponding rectangle.

To the best of our knowledge, general \Stabbing has not been studied,
although it is a natural problem.  Finke et
al.~\cite{finke08-batch-processing} consider the special case of the
problem where the left sides of all input rectangles lie on
the~${y}$-axis.  They derive the problem from a practical application
in the area of batch processing and give a polynomial time algorithm
that solves this special case of \Stabbing to optimality.  Das et
al.~\cite{dfksvw-agmmnp-Algorithmica18} describe an application of
\Stabbing in geometric network design.  They obtain a constant-factor
approximation for a slight generalization of the special case of Finke
et al.\ in which rectangles are only constrained to \emph{intersect}
the~${y}$-axis.  This result constitutes the key step for
an~${\bigOh(\log n)}$-approximation algorithm to the \GMMN problem.

We also consider the following variant of our problem, which we call
\ConstrainedStabbing. Here, the input additionally consists of a
set~${F}$ of horizontal line segments of which any solution~${S}$ must
be a subset. We will also consider the unweighted variant of
\ConstrainedStabbing, called \CardinalityStabbing, where the objective
consists in minimizing the \emph{number} of segments.

\paragraph{Related Work.}
\Stabbing can be interpreted as a weighted geometric set cover problem
where the rectangles play the role of the elements, the potential line
segments correspond to the sets and a segment~${s}$ ``contains'' a
rectangle~${r}$ if~${s}$ stabs~${r}$.  The weight of a segment~${s}$
equals its length~${\len{s}}$.  \problemName{Set Cover} is one of the
classical \NP-hard problems.  The greedy algorithm yields
a~${\ln n}$-approximation (where~${n}$ is the number of elements) and
this is known to be the best possible approximation ratio for the
problem unless
$\runtimeclass{P}=\NP$~\cite{feige98-setcover,DinurS14:set-cover-pvsnp}.
It is an important field of computational geometry to surpass the
lower bound known for general \problemName{Set Cover} in geometric
settings.  In their seminal work, Br\"onniman and
Goodrich~\cite{bronnimann-goodrich95-setcover-vcdimension} gave
an~${\bigOh(\log\opt)}$-approximation algorithm for \emph{unweighted}
\problemName{Set Cover}, where~${\opt}$ is the size of an optimum
solution, for the case when the underlying
\emph{VC-dimension}\footnote{Informally, the~VC-dimension of a set
  cover instance~$(\universe,\sets)$ is the size of a largest
  subset~$X\subseteq\universe$ such that~$X$ \emph{induces} in~$\sets$
  the set cover instance~$(X,2^X)$.}  is constant.  This holds in many
geometric settings.  Numerous subsequent works have improved upon this
result in specific geometric settings.  For example, Aronov et
al.~\cite{aronov10-eps-nets-rectangles} obtained
an~${\bigOh(\log\log\opt)}$-approximation algorithm for the problem of
piercing a set of axis-aligned rectangles with the minimum number of
points (\problemName{Hitting Set} for axis-aligned rectangles) by
means of so-called~\emph{${\epsilon}$-nets}.  Mustafa and Ray
\cite{mustafa-ray10-geometric-hitting-set} obtained a PTAS for the
case of piercing pseudo-disks by points.  A limitation of these
algorithms is that they only apply to \emph{unweighted} geometric
\problemName{Set Cover}; hence, we cannot apply them directly to our
problem. In a break-through, Varadarajan
\cite{varadarajan10-quasi-uniform-sampling} developed a new technique,
called \emph{quasi-uniform sampling}, that gives sub-logarithmic
approximation algorithms for a number of \emph{weighted} geometric set
cover problems (such as covering points with weighted fat triangles or
weighted disks).  Subsequently, Chan et
al.~\cite{chan-etal12-weighted-geom-sc} generalized Varadarajan's
idea.  They showed that quasi-uniform sampling yields a
sub-logarithmic performance if the underlying instances have low
\emphasizedAlreadyInAbstract{shallow-cell complexity}.  Bansal and
Pruhs \cite{bansal-pruhs14-geometry-scheduling} presented an
interesting application of Varadarajan's technique.  They reduced a
large class of scheduling problems to a particular geometric set cover
problem for anchored rectangles and obtained a constant-factor
approximation via quasi-uniform sampling.  Recently, Chan and
Grant~\cite{chan-grant-exactAlgos} and Mustafa et
al.~\cite{Mustafa2015} settled the \APX-hardness status of all natural
weighted geometric \problemName{Set Cover} problems where the elements
to be covered are points in the plane or space.

Gaur et al.~\cite{gaur-etal02-stabbing-rects-lines} considered the
problem of stabbing a set of axis-aligned rectangles by a minimum
number of axis-aligned \emph{lines}.  They obtain an elegant
$2$-approximation algorithm for this \NP-hard problem by rounding the standard LP-relaxation.  Kovaleva and
Spieksma~\cite{kovaleva-etal06-weighted-stabbing} considered a
generalization of this problem involving weights and demands.  They
obtained a constant-factor approximation for the problem.  Even et
al.~\cite{even-etal08-capacitated-stabbing} considered a
\emph{capacitated} variant of the problem in arbitrary dimension.
They obtained approximation ratios that depend linearly on the
dimension and extended these results to approximate certain lot-sizing
inventory problems.  Giannopoulos et
al.~\cite{giannopoulos-etal13-stabbing-fpt} investigated the
fixed-parameter tractability of the problem where given translated
copies of an object are to be stabbed by a minimum number of lines
(which is also the parameter).  Among others, they showed that the
problem is~${\runtimeclass{W}[1]}$-hard for unit-squares but becomes
\FPT if the squares are disjoint.

\paragraph{Our Contribution.}
We are the first to investigate \Stabbing in its general
form. (Previous works considered only special cases of the problem.)
We examine the complexity and the approximability of the problem.

First, we rule out the possibility of efficient exact algorithms by
showing that \Stabbing is \NP-hard; see
Section~\ref{sec:stab-np}.  \ConstrainedStabbing and
\CardinalityStabbing turn out to be even \APX-hard; see
Section~\ref{sec:stab-apx}.

Another negative result is that \Stabbing instances can have high
shallow-cell complexity so that a direct application of the
quasi-uniform sampling method yields only the same logarithmic bound
as for arbitrary set cover instances.

Our main result is a constant-factor approximation algorithm for
\Stabbing which is based on the following three ideas.  First, we show
a simple decomposition lemma that implies a constant-factor
approximation for (general) set cover instances whose set family can
be decomposed into two disjoint sub-families each of which admits a
constant-factor approximation.  
Second, we show that \Stabbing
instances whose segments have a special \emph{laminar} structure have
low enough shallow-cell complexity so that they admit a
constant-factor approximation by quasi-uniform sampling.  Third, we
show that an arbitrary instance can be transformed in such a way that
it can be decomposed into two disjoint laminar families.  Together
with the decomposition lemma, this establishes the constant-factor
approximation.

Another (this time more obvious) application of the decomposition
lemma gives also a constant-factor approximation for the variant of
\Stabbing where we allow horizontal and vertical stabbing segments.
Also in this case, a direct application of quasi-uniform sampling
gives only a logarithmic bound as there are laminar families of
horizontal \emph{and} vertical segments that have high shallow-cell
complexity.  This and two further applications of the decomposition
lemma are sketched in Section~\ref{sec:applications}.

The above results provide two natural examples for the fact that the
property of having low shallow-cell complexity is \emph{not} closed
under the union of the set families.  In spite of this,
constant-factor approximations are still possible.  Our results also
show that the representation as a union of low-complexity families may
not be obvious at first glance.  We therefore hope that our approach
helps to extend the reach of quasi-uniform sampling beyond the concept
of low shallow-cell complexity also in other settings.  Our results
for \Stabbing may also lead to new insights for other related
geometric problems such as the \GMMN
problem~\cite{dfksvw-agmmnp-Algorithmica18}.

As a side remark, we first explore the relationship of
\problemName{Stabbing} to well-studied geometric set cover (or
equivalently hitting set) problems; see
Appendix~\ref{sec:relat-probs}.  We show that \Stabbing can be seen as
(weighted) \problemName{Hitting Set} for axis-aligned boxes in 
three dimensions.  
This immediately implies an~${\bigOh(\log\log
	n)}$-approximation algorithm for \CardinalityStabbing, the
unweighted variant.  The embedding does not yield a sub-logarithmic
performance for \Stabbing, however.  A similar embedding is not
possible in two dimensions: There are set cover instances that can be
realized as instances of our problem but not as instances of
\problemName{Hitting Set} for axis-aligned rectangles.  We also show
that natural greedy approaches for \Stabbing fail to beat the
logarithmic bound.

\section{NP-Hardness of \Stabbing}
\label{sec:stab-np}

To show that \Stabbing is \NP-hard,
we reduce from \problemName{Planar Vertex~Cover}: Given a planar
graph~$G$ and an integer~$k$, decide whether $G$ has a vertex cover of
size at most~$k$.  This problem is
\NP-hard~\cite{Garey-1974-planar-vc-np-complete}.

\begin{theorem}
  \Stabbing is \NP-hard, even for interior-disjoint rectangles.
\end{theorem}

\begin{figure}[t]
  \begin{subfigure}[b]{.4\textwidth}
    \centering
    \includegraphics[page=1]{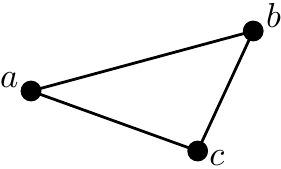}
    \caption{a \problemName{Planar Vertex~Cover} instance}
    \label{fig:np-G}
  \end{subfigure}
  \hfill
  \begin{subfigure}[b]{.28\textwidth}
    \centering
    \includegraphics[page=2]{npGadgets1}
    \caption{a visibility representation}
    \label{fig:np-VisRep}
  \end{subfigure}
  \hfill
  \begin{subfigure}[b]{.23\textwidth}
    \centering
    \includegraphics[page=3]{npGadgets1}
    \caption{no coinciding edges}
    \label{fig:np-ModVisRep}
  \end{subfigure}

  \caption{Obtaining a visibility representation from a
    \problemName{Planar Vertex~Cover} instance.}
  \label{fig:vtx-cv-instance}
\end{figure}

Let~$G=(V,E)$ be a planar graph with $n$ vertices, and let~$k$ be a
positive integer.  Our reduction will map~$G$ to a set~$R$ of
rectangles and $k$ to another integer~$k^\star$ such that $(G,k)$ is a
yes-instance of \problemName{Planar Vertex Cover} if and only if
$(R,k^\star)$ is a yes-instance of \Stabbing.
Consider a \emph{visibility representation} of $G$, which 
represents the vertices of~$G$ by non-overlapping vertical line
segments (called \emph{vertex segments}), and each edge of~$G$ 
by a horizontal line segment (called \emph{edge segment}) that touches
the vertex segments of its endpoints; see
\figs~\ref{fig:np-G} and~\ref{fig:np-VisRep}.  Any planar graph admits
a visibility representation on a grid of size $O(n) \times O(n)$,
which can be found in polynomial time~\cite{Lin-SIAM04-visRep}.
We compute such a visibility representation for~${G}$. 
Then we stretch the vertex segments and vertically shift the edge segments
so that no two edge segments coincide (on a vertex segment);
see \fig~\ref{fig:np-ModVisRep}.  The height of the visibility
representation remains linear.

In the next step, we create a \Stabbing{} instance based on this visibility representation,
using the edge segments and vertex segments as indication
for where to put our rectangles. 
All rectangles will be interior-disjoint, 
have positive area 
and lie on an integer grid that we obtain by scaling the visibility
representation by a sufficiently large factor (linear in~$n$). 
A vertex segment 
will intersect~${O(n)}$ rectangles (above each other, since they are disjoint), and each
rectangle will have width~${O(n)}$.
The precise number of rectangles and their sizes will depend on the constraints formulated below. 
Our construction will be polynomial in~$n$.

For each edge~${e}$ in $G$, we introduce an \emph{edge gadget}
$\ReG{e}$, which is a rectangle that we placed such that it is stabbed
by the edge segment of~$e$ in the visibility representation.

\begin{figure}[tb]
  \begin{subfigure}[t]{.38\textwidth}
    \centering
    \includegraphics[page=1]{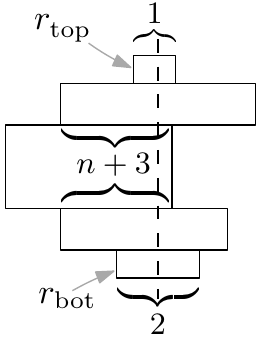}
    \caption{all rectangles of~$\RvG{v}$ are intersected by the
      (dashed) vertex segment of~$v$}%
    \label{fig:vtxGadget}
  \end{subfigure}
  \hfill
  \begin{subfigure}[t]{.23\textwidth}
    \centering
    \includegraphics[page=2]{npGadgets2}
    \caption{$\RvG{v}$ stabbed by~$\Sact{v}$}
    \label{fig:vtxActGadget}
  \end{subfigure}
  \hfill 
  \begin{subfigure}[t]{.23\textwidth}
    \centering
    \includegraphics[page=3]{npGadgets2}
    \caption{$\RvG{v}$ stabbed by~$\Sina{v}$}
    \label{fig:vtxInaGadget}
  \end{subfigure}
  \caption{The vertex gadget~$\RvG{v}$ of vertex~$v$}
\end{figure}

For each vertex~$v$ in $G$, we introduce a \emph{vertex gadget} $\RvG{v}$
as shown in \fig~\ref{fig:vtxGadget}.  It consists of an odd number of
rectangles that are (vertically) stabbed by the vertex segment
of~$v$ in the visibility representation.
Any two neighboring rectangles share a horizontal line segment.
Its length is exactly~${n+3}$ if neither of the rectangles is the
top-most rectangle~${\rTop}$ or the bottom-most rectangle~${\rBot}$.
Otherwise, the intersection length equals the width of the respective
rectangle~${\rTop}$ or~${\rBot}$.
We set the widths of~$\rTop$ and~$\rBot$ to~$1$ and $2$, respectively.
A vertex gadget~${\RvG{v}}$ is called \emph{incident} to an edge
gadget~${\ReG{e}}$ if~${v}$ is incident to~${e}$.

Before we describe the gadgets and their relation to each other in
more detail, we construct, in two steps, a set~$S^v$ of line segments
for each vertex gadget~$\RvG{v}$.
First, let $S^v$ be the set of line segments that correspond to the 
top and bottom edges of the rectangles in~$\RvG{v}$.  Second, replace
each pair of overlapping line segments in~$S^v$ by its union.
Then number the line segments in~$S^v$ from top to bottom starting
with~$1$.  Let~$\Sina{v}$ be the set of the odd-numbered line
segments, and let~${\Sact{v}}$ be the set of the even-numbered ones;
see \figs~\ref{fig:vtxActGadget} and~\ref{fig:vtxInaGadget}.
By construction, ${\Sact{v}}$ and~${\Sina{v}}$ are feasible stabbings
for~${\RvG{v}}$.  Furthermore,~${|\Sina{v}| = |\Sact{v}|}$
as~$|\RvG{v}|$ is odd and, hence, $|S^v|$ is even.
Given the difference in the widths of~${\rTop}$ and~${\rBot}$,
we have that $\lenAct{v} = \lenIna{v} + 1$.
Note that this equation holds regardless of the widths of the
rectangles in~${\RvG{v}\setminus\{\rTop,\rBot\}}$.

The rectangles of all gadgets together form a \Stabbing
instance~${R}$. 
They meet two further constraints: 
First, no two rectangles of different vertex gadgets intersect. We
can achieve this by scaling the visibility representation by an
appropriate linear factor. 
Second, each edge
gadget~${\ReG{e}}$ intersects exactly two rectangles, one of its
incident left vertex gadgets,~$\RvG{v}$, and one of its incident
right vertex gadgets,~$\RvG{u}$. 
The top edge of~${\ReG{e}}$ touches
a segment of~${\Sact{v}}$ and the bottom edge of~${\ReG{e}}$ touches a
segment of~${\Sact{u}}$. The length of each of the two intersections
is exactly~${n+3}$; see \fig~\ref{fig:np-R}. 
Thus, we have~${|\RvG{v}|=O(\deg(v))=O(n)}$.

\begin{figure}[tb]
  \centering
  \includegraphics{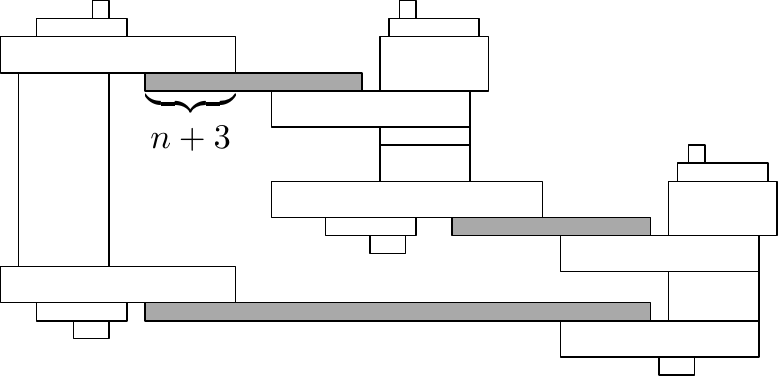}
  \caption{The \Stabbing instance that encodes the \problemName{Planar
      Vertex Cover} instance of \fig~\ref{fig:vtx-cv-instance}; edge
    gadgets are shaded gray.}
  \label{fig:np-R}
\end{figure}

Let~${S}$ be a feasible solution to the instance~${R}$.  We call a vertex
gadget~${\RvG{v}}$ \emph{active} in~$S$ 
if~${\{ s \cap \bigcup \RvG{v} \mid s \in S\} = \Sact{v}}$, and
\emph{inactive} in~$S$ if ${\{ s \cap \bigcup \RvG{v} \mid
  s \in S\} =  \Sina{v}}$.
We will see that in any optimum solution each vertex gadget is
either active or inactive.  Furthermore, we will establish a direct
correspondence between the \problemName{Planar Vertex Cover}
instance~$G$ and the \Stabbing instance~$R$: Every optimum solution
to~${R}$ \emph{covers} each edge gadget by an active vertex gadget
while minimizing the number of active vertex gadgets.

Let~${\opt_G}$ denote the size of a minimum vertex cover for~${G}$,
let~${\opt_R}$ denote the length of an optimum solution to~${R}$,
let~${\w(r)}$ denote the width of a rectangle~${r}$, and finally
let~$\fixC = \sum_{e\in E} \left(\w(\ReG{e}) - n - 3\right) +
\sum_{v\in V} \lenIna{v}$.
We now prove that $\opt_G \le k$ if and only if $\opt_R \le \fixC+k$.
We show the two directions separately. 

\newcounter{redLR}
\setcounter{redLR}{\thetheorem}
\newcommand{\lemmaRedLR}{%
  $\opt_G \le k$ implies that $\opt_R \le \fixC + k$.}
\begin{lemma}
  \label{lem:redLR}
  \lemmaRedLR
\end{lemma}
\begin{proof}
	Given a vertex cover of size~${k'\le k}$, we set all vertex gadgets
	that correspond to vertices in the vertex cover to active and all
	the other ones to inactive. Then for each edge gadget~${\ReG{e}}$, at
	least one incident vertex gadget is active, say~${\RvG{v}}$.  By our
	construction of~${R}$, there is a line segment~${s}$ in~${\Sact{v}}$
	with~${\len{ s \cap \ReG{e}} = n+3}$.  We increase the length
	of~${s}$ by~${\w(\ReG{e}) - n-3}$ so that~${\ReG{e}}$ is stabbed.
	Hence, we obtained a feasible solution to~$R$.  Recall that there
	are~${k'}$ active vertex gadgets and that, for each vertex~${v}$, we
	have~${\lenAct{v} = \lenIna{v} + 1}$.  Thus, the total length of our
	solution is
	\[\sum\limits_{v\in V} \lenIna{v} + k' + \sum\limits_{e \in E}
	\left(\w(\ReG{e}) - n - 3\right) \le \fixC + k' \le \fixC + k\] and
	the lemma follows.
\end{proof}

Next we show the other direction, which
is more challenging.  Consider an optimum solution~${\Sopt}$ to~${R}$
and choose~$k \le n$ such that~${\opt_R \le \fixC + k}$ is
satisfied.  Let~${\RvG{v}}$ be any vertex gadget, and let~${\rTop}$
and~${\rBot}$ be its top- and bottom-most rectangles.
Furthermore, let~${\SvOpt{v} = \{ s \cap \bigcup \RvG{v} \mid s \in \Sopt\}}$.  In
the following steps, we prove that~${\SvOpt{v}}$ equals
either~${\Sina{v}}$ or~${\Sact{v}}$.

First, we make two observations.
We transform~${\Sopt}$ as follows without increasing its total length.
Let~${s\in \Sopt}$ be a line segment stabbing~${\rTop}$. If~${s}$ 
stabs only~${\rTop}$, then we move~${s}$ to the top edge
of~${\rTop}$.  If~${s}$ also stabs other rectangles, then one
of these rectangles must touch~${\rTop}$ (otherwise we could
split~${s}$ and shrink its subsegments, contradicting
optimality). Note that the only rectangle touching~${\rTop}$ lies below it
and belongs to~${\RvG{v}}$.  A similar argument holds
for~${\rBot}$.
\begin{observation}\label{obs:TopBot}
	Without loss of generality, it holds that:
	\begin{listRoman}
		\item\label{obs:Top} Any segment in~${\Sopt}$ stabbing~${\rTop}$
		either stabs~${\rTop}$ through its top edge, or also stabs
		another rectangle in~${\RvG{v}}$.
		\item\label{obs:Bot} The same holds for the rectangle~${\rBot}$ and
		its bottom edge.
	\end{listRoman}
\end{observation}

\begin{observation}\label{obs:shortest}
	Every line segment in~${\SvOpt{v}}$ that does not stab a rectangle
	in~${\{\rTop, \rBot\}}$ has length at least~${n+3}$.
\end{observation}

Note that Observation~\ref{obs:shortest} also holds for line
segments that stab only rectangles belonging to edge gadgets as those
rectangles have length at least~$n+3$.

\newcounter{containment}
\setcounter{containment}{\thetheorem}
\newcommand{\lemmaContainment}{%
  If~${\Sina{v} \not\subseteq \SvOpt{v}}$ and~${\Sact{v} \not\subseteq
    \SvOpt{v}}$, then~${\vOpt{v} > \lenAct{v} + n}$.}
\begin{lemma}
  \label{lem:containment}
  \lemmaContainment
\end{lemma}
\begin{proof}
	We say that a pair of rectangles is \emph{stabbed} by a line
	segment if the line segment stabs both rectangles.  Let~${P}$
	be a maximum-cardinality set of rectangle pairs
	of~${\RvG{v}}$ where each
	pair is stabbed by a line segment in~${\SvOpt{v}}$ and each
	rectangle appears in at most one pair.
	For~${\Sina{v}}$ and~${\Sact{v}}$, such a maximum-cardinality
	set of pairs is unique and excludes exactly one
	rectangle, namely~$\rTop$ or~$\rBot$, respectively.
	
	Now, as~${\RvG{v}}$ contains an odd
	number of rectangles, the number of rectangles not in~${P}$ is
	odd and at least one.  If there is exactly one
	rectangle not in~${P}$,	this rectangle is different
	from~${\rTop}$ and~${\rBot}$, as otherwise
	Observation~\ref{obs:TopBot} would yield~${\Sina{v} \subseteq
		\SvOpt{v}}$ or~${\Sact{v} \subseteq \SvOpt{v}}$; 
	a contradiction since~${\Sina{v} \not\subseteq \SvOpt{v}}$ and~${\Sact{v} \not\subseteq \SvOpt{v}}$.
	If there are at least three rectangles not in~${P}$, then one among them is different from~${\rTop}$ and~${\rBot}$.
	Hence,
	in both cases, there is at least one rectangle~${r'}$ not in~${P}$
	that is different from~${\rTop}$ and~${\rBot}$.
	
	Thus,~${\SvOpt{v}}$ contains a line segment that stabs~${r'}$ and
	that does not stab any other rectangle pair in~${P}$. The line
	segment is not shorter than~${\w(r')}$.  Furthermore, for each
	pair~${(r_1,r_2) \in P}$,~${\SvOpt{v}}$ contains a line segment of
	length
	${\w(r_1)+\w(r_2) - \w(r_1 \cap r_2)}$ that stabs~${r_1}$
	and~${r_2}$.  Putting things together, we bound~${\vOpt{v}}$ from
	below by
	\begin{equation}
	\sum_{r \in \RvG{v}} \w(r) - \sum_{(r_1,r_2) \in P} \w(r_1 \cap r_2)\formulaPunctuationSpace.\hfill\label{eq:lower-bound}
	\end{equation}
	The first sum is independent of~${\SvOpt{v}}$.  Thus,
	bound~\eqref{eq:lower-bound} is minimized by maximizing the second
	sum.  Let's examine the value of~${\w(r_1 \cap r_2)}$ for
	various pairs~${(r_1,r_2)}$.  For
	the unique pair containing~${\rTop}$, the value is~${1}$, for the
	unique pair containing~${\rBot}$, it is~${2}$. For all the other
	pairs, by construction it is~${n+3}$.  Thus, the second sum is maximized
	when~${\rTop}$ is the only rectangle not in~${P}$.  This is exactly
	the case for~${\Sina{v}}$. As~${\Sina{v}}$ contains one line segment
	for each pair and one line segment for~${\rTop}$, and each line
	segment is only as long as necessary,~${\lenIna{v}}$ reaches
	bound~\eqref{eq:lower-bound} and is consequently optimal.
	
	Due to the assumption of the lemma, there is a
	rectangle~$r'$ that is not in~$P$.  Note that the second sum
	is maximized if~$r'$ is the \emph{only} rectangle not in~$P$.  
	Compared to the optimal situation when $\SvOpt{v}=\Sina{v}$,
	the value of the sum changes by~${1 - (n+3)}$: replace the
	pair with~$r'$ in $\Sina{v}$ by the pair with~$\rTop$.
	Consequently, under the assumption of the lemma,
	bound~\eqref{eq:lower-bound} for~${\SvOpt{v}}$ is at 
	least~${\lenIna{v} + n + 2 = \lenAct{v} + n + 1}$, and the claim
	follows.
\end{proof}

\newpage
\newcounter{actOrExp}
\setcounter{actOrExp}{\thetheorem}
\newcommand{\lemmaActOrExp}{%
  Exactly one of the following three statements holds:
  \begin{listRoman}
  \item~${\SvOpt{v}=\Sina{v}}$, or
  \item~${\SvOpt{v}=\Sact{v}}$, or
  \item~${\vOpt{v} > \lenIna{v}+n}$.
  \end{listRoman}
}
\begin{lemma}
  \label{lem:actOrExp}
  \lemmaActOrExp
\end{lemma}
\begin{proof}
	Suppose that~${\Sina{v}}$ is a proper subset of~${\SvOpt{v}}$ for some
	vertex~${v}$, and let~${s\in\SvOpt{v}\setminus\Sina{v}}$.  
	Consider the 
	segment~${s'\in\Sopt}$ that \enquote{induces}~$s$, that is,~${s=s'\cap \bigcup \RvG{v}}$. 
	If~${s}$ stabs
	only a rectangle in~${\{\rTop,\rBot\}}$, then, by
	Observation~\ref{obs:TopBot},~${s'}$ stabs no other rectangle
	in~${R}$.  Hence, we can safely remove~${s'}$ from~${\Sopt}$,
	as~${\rTop}$ and~${\rBot}$ are already stabbed in~${\Sina{v}}$; a
	contradiction to the optimality of~${\Sopt}$.  Consequently,~${s}$
	must stab a rectangle in~${\RvG{v}\setminus \{\rTop, \rBot\}}$.  By
	Observation~\ref{obs:shortest}, we get
	\[{\vOpt{v} \ge \lenIna{v} + n + 3 > \lenIna{v} +
		n}\formulaPunctuationSpace.\]
	The same holds for~${\Sact{v} \subsetneq \SvOpt{v}}$.
	By Lemma~\ref{lem:containment}, this yields the claim.
\end{proof}

Now, we show that~${\Sopt}$ forces each vertex gadget to
be either active or inactive. 

\begin{lemma}\label{lem:actOrIna}
  In~${\Sopt}$, each vertex gadget is either active or inactive.
\end{lemma}
\begin{proof}
  Suppose that there is a vertex gadget~${\RvG{u}}$ that is neither
  active nor inactive in~${\Sopt}$.  This
  implies~${\opt_R > \fixC + n}$ and contradicts our previous
  assumption~${\opt_R \le \fixC + k \le \fixC + n}$.
	
  To this end, we give a lower bound on~${\opt_R}$.
  Since~${\RvG{u}}$ is neither active nor inactive,~${\SvOpt{u}>
    \lenIna{u}+n}$ by Lemma~\ref{lem:actOrExp}. Thus,
  $\sum_{v \in V} \len{\SvOpt{v}} > \sum_{v\in V}
  \lenIna{v} + n\formulaPunctuationSpace.$
  Let~${\Sopt^{\textrm{out}}}$ be the set of all segment fragments
  of~${\Sopt}$ lying outside of~${\bigcup_{v\in V} \SvOpt{v}}$.
  Each edge gadget~${\ReG{v}}$ contains a segment fragment
  from~${\Sopt^{\textrm{out}}}$ of length at least~${\w(\ReG{v}) - n -
    3}$, since by construction it can share a line segment with only
  one of its incident vertex gadgets.  Since all edge gadgets are
  interior-disjoint, we have
  $\len{\Sopt^{\textrm{out}}} \ge \sum_{e\in E} \w(\ReG{v}) - n
  - 3$. Hence,
  \begin{alignat*}{1}
    \opt_R \; \ge & \; \len{\Sopt^{\textrm{out}}} + \sum_{v \in
      V} \len{\SvOpt{v}} \\
    > & \sum\limits_{e\in E} \left(\w(\ReG{e}) - n - 3\right) +
    \sum\limits_{v\in V} \lenIna{v} + n \;=\; c + n \formulaPunctuationSpace. \hskip97px
    \qedhere
  \end{alignat*}
\end{proof}

\begin{lemma}\label{lem:gadgetCover}
  For each edge gadget, one of its incident vertex gadgets is active
  in~${\Sopt}$.
\end{lemma}
\begin{proof}
  Suppose that for an edge gadget~${\ReG{e}}$ both vertex gadgets
  are not active in~${\Sopt}$.  By Lemma~\ref{lem:actOrIna}, they are
  inactive.  Without loss of generality, the line segment~${s}$
  stabbing~${\ReG{e}}$ lies on the top or bottom edge
  of~${\ReG{e}}$. 
  Then~${s}$ intersects
  a vertex gadget to the left or right, say~${\RvG{v}}$, and hence~${\SvOpt{v}
    \not= \Sina{v}}$ and~${\SvOpt{v} \not= \Sact{v}}$. A contradiction
  given Lemma~\ref{lem:actOrIna}.
\end{proof}

\newcounter{gadgetCost}
\setcounter{gadgetCost}{\thetheorem}
\newcommand{\lemmaGadgetCost}{%
  $\opt_R = \fixC + k'$, where~${k'}$ is the number of active vertex
  gadgets in~${\Sopt}$.
}
\begin{lemma}
  \label{lem:gadgetCost}
  \lemmaGadgetCost
\end{lemma}
\begin{proof}
	Consider any edge gadget~${\ReG{e}}$.  It is stabbed by only one
	line segment~${s}$, and, without loss of generality, the line
	segment~${s}$ lies on the top or bottom edge
	of~${\ReG{e}}$. Thus, it intersects a vertex gadget~${\RvG{v}}$ on a
	rectangle~${r}$. Then~${\RvG{v}\ne \Sina{v}}$ and~${\RvG{v}}$ is
	active according to Lemma~\ref{lem:actOrIna}.  By our construction
	of~${R}$, there is exactly one segment in~${\Sact{v}}$
	intersecting~${\ReG{e}}$, which also stabs~${r}$. Hence, this
	segment is a subsegment of~${s}$ and we have
	\[{\len{s} = \w(r) + \w(\ReG{e}) - \w(r \cap \ReG{e}) = \w(r) +
		\w(\ReG{e}) - n-3}\formulaPunctuationSpace.\]
	Thus, by Lemma~\ref{lem:actOrIna} and~${\lenAct{v}=\lenIna{v}+1}$,
	\begin{equation*}
	\opt_R~~=\enspace~ \sum\limits_{e\in E} \left(\w(\ReG{e}) -
	n - 3\right)~~ + \sum\limits_{v\in V} \lenIna{v}~~ +~~ k'
	~~=\enspace~ c+ k'
	\end{equation*}
	where~${k'}$ is the number of active vertex gadgets in~${\Sopt}$.
\end{proof}

Given~${\Sopt}$, we put exactly those vertices in the vertex cover
whose vertex gadgets are active.  By Lemma~\ref{lem:gadgetCover}, this
yields a vertex cover of~${G}$.  By Lemma~\ref{lem:gadgetCost},
the size of the vertex cover is exactly~${\opt_R - \fixC}$, which is
bounded from above by~${k}$ given that~${\opt_R\le c + k}$.
\begin{lemma}\label{lem:redRL}
  $\opt_R \le \fixC + k$ implies that $\opt_G \le k$.
\end{lemma}

With Lemmas~\ref{lem:redLR} and~\ref{lem:redRL}, 
we conclude that \Stabbing is \NP-hard. 

\section{APX-Hardness of Cardinality and Constrained Stabbing}
\label{sec:stab-apx}

In this section, we consider \CardinalityStabbing and
\ConstrainedStabbing.  The latter is the variant of \Stabbing, where
the solution is constrained to be some subset of a given set of line
segments.  By reducing a restricted \APX-hard variant of
\problemName{Set Cover} to these problems, we show that neither
\CardinalityStabbing nor \ConstrainedStabbing admits a PTAS.  The
following lemma follows directly from Definition~1.2 and Lemma~1.3 by
Grant and Chan \cite{chan-grant-exactAlgos}.

\begin{lemma}[Grant and Chan \cite{chan-grant-exactAlgos}]
	\SPSC is \APX-hard, where \SPSC is defined as unweighted
	\problemName{Set Cover} with the following properties: The input is
	a family~${\mathcal{S}}$ of subsets of a universe~${\mathcal{U} = A
		\cup W \cup X \cup Y \cup Z}$ that comprises the disjoint sets
	(which are part of the input)~${A = \{a_1, \dots, a_n\}}$,~${W =
		\{w_1, \dots, w_m\}}$,~${X = \{x_1, \dots, x_m\}}$,~${Y = \{y_1,
		\dots, y_m\}}$, and~${Z = \{z_1, \dots, z_m\}}$ where~${2n = 3m}$.
	The family~${\mathcal{S}}$ consists of~${5m}$ sets and satisfies the
	following two conditions:
	\begin{itemize}
		\item For every~${t}$ with~${1\le t \le m}$, there are
		integers~${i}$ and~${j}$ with~${1 \le i < j < k \le n}$ such
		that~${\mathcal{S}}$ contains the sets~${\{a_i, w_t\}}$,~${\{w_t,
			x_t\}}$,~${\{a_j, x_t, y_t\}}$,~${\{y_t, z_t\}}$, and~${\{a_k,
			z_t\}}$.  (See \fig~\ref{fig:specialExample}.)
		\item For every~${i}$ with~${1 \le i \le n}$, the element~${a_i}$ is
		in exactly two sets in~${\mathcal{S}}$.
	\end{itemize}
\end{lemma}

\begin{figure}[tb]
	\begin{minipage}{0.45\linewidth}
		\centering \subcaptionbox{\label{fig:specialExample}%
			For~${1 \le t \le m}$, there are five sets. Each element appears
			in exactly two sets.}%
		{\quad\includegraphics{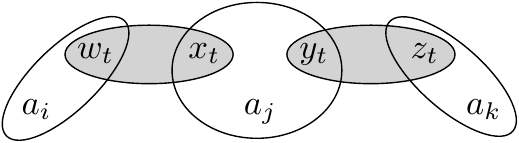}\quad}%
		\vspace*{4ex}%
		
		\subcaptionbox{\label{fig:allEncoding}%
			Distinct rectangles~${a_1, \dots, a_n}$ ($a_1$ is thick) have
			intersection~${I}$ (shaded) which is subdivided into areas~${1,
				\dots, m}$.}%
		{\quad\includegraphics[page=6]{apxHardnessForGivenSegments}\quad}%
	\end{minipage}%
	\hfill
	\begin{minipage}{0.5\linewidth}
		\centering \subcaptionbox{\label{fig:encoding}%
			For~${1 \le t \le m}$, there are five line segments (thick gray
			horizontal line segments). Observe that each line segment stabs
			exactly one of~${\{a_i, w_t\}}$,~${\{w_t, x_t\}}$,~${\{a_j, x_t,
				y_t\}}$,~${\{y_t, z_t\}}$, and~${\{a_k, z_t\}}$.%
		}%
		{\qquad\includegraphics[page=4]{apxHardnessForGivenSegments}\qquad\vspace{7pt}}%
	\end{minipage}
	\caption{Encoding of \SPSC via \Stabbing.}%
\end{figure}

We begin with \CardinalityStabbing.

\begin{theorem}
	\label{thm:apx-card-stab}
	\CardinalityStabbing is \APX-hard.
\end{theorem}
\begin{proof}
	Given a \SPSC instance~${(\mathcal{U}, \mathcal{S})}$
	with~${\mathcal{U} = A \cup W \cup X \cup Y \cup Z}$, we efficiently
	encode it as a \Stabbing instance by creating a rectangle for each
	element of the universe~${\mathcal{U}}$ and adding a line segment
	for each set of~${\mathcal{S}}$. We will achieve the property that a
	line segment corresponding to a set~${s}$ stabs exactly those
	rectangles that correspond to the elements of~${s}$.
	
	Let~$n = |A|$ and~$m=|W|$ (recall~${2n = 3m}$).  To
	encode~${(\mathcal{U}, \mathcal{S})}$, we place~${n}$ rectangles of
	equal size on one spot and then shift them one by one to the right
	such that all rectangles are distinct and their total
	intersection~${\intAs}$ is not empty.  For~${1\le i \le n}$,
	the~${i}$\thSuffix rectangle from the left corresponds to
	element~${a_i}$; see \fig~\ref{fig:allEncoding}.
	
	Next, we horizontally subdivide the intersection~${\intAs}$
	into~${m}$ areas.  For~${1\le t\le m}$, we place four thin
	rectangles inside the~${t}$\thSuffix area from the top such that the
	vertical projections of the rectangles intersect sequentially and
	only pairwise as in \fig~\ref{fig:encoding}.  From top to bottom
	they correspond to~${w_t}$,~${x_t}$,~${y_t}$, and~${z_t}$.
	
	Now, we show that our rectangle configuration allows a feasible set
	of line segments that corresponds to~${\mathcal{S}}$.  Recall the
	definition of \SPSC. For~${1 \le t \le m}$, the input contains the
	sets~${\{a_i, w_t\}}$,~${\{w_t, x_t\}}$,~${\{a_j, x_t,
		y_t\}}$,~${\{y_t, z_t\}}$ and~${\{a_k, z_t\}}$.  Consider the
	first set.  As the rectangle~${w_t}$ is inside the
	rectangle~${a_i}$, we can stab both with one line segment.  We can
	even stab them exclusively if we place the line segment above the
	rectangle~${x_t}$ and do not leave the rectangle~${a_i}$ (note that
	no rectangle of~${A}$ is contained in another one).  Hence, the line
	segment corresponds to the set~${\{a_i, w_t\}}$.  With a similar
	discussion, we can find line segments that correspond to the
	sets~${\{w_t, x_t\}}$,~${\{a_j, x_t, y_t\}}$,~${\{y_t, z_t\}}$
	and~${\{a_k, z_t\}}$, respectively; see \fig~\ref{fig:encoding}.
	
	Since the objective is to minimize the cardinality of line segments,
	there is a cost-preserving correspondence between solutions to the
	\SPSC instance and the generated \Stabbing instance.  Hence,
	\CardinalityStabbing is \APX-hard.
\end{proof}

Next, we show that there is an
\emph{$L$-reduction}~\cite{papadimitriou-approximation} from \SPSC
onto \ConstrainedStabbing (where we minimize the total segment
length).  This implies \APX-hardness
\cite{papadimitriou-approximation}.

\begin{theorem}
	\label{thm:apx-constr-stab}
	\ConstrainedStabbing is \APX-hard.
\end{theorem}
\begin{proof}
	For an~${L}$-reduction, it suffices to find two constants~${\alpha}$
	and~${\beta}$ such that for every \SPSC instance~${\xSSC}$ it holds:
	\begin{listArabic}
		\item We can efficiently construct a \ConstrainedStabbing
		instance~${\xStab}$ with \[\coptStab(\xStab) \le \alpha \cdot
		\coptSSC(\xSSC)\] \label{Lred:one} where~$\coptStab(\xStab)$ is
		the total length of an optimum solution to~$\xStab$
		and~$\coptSSC(\xSSC)$ is the cardinality of a minimum set cover
		to~$\xSSC$.
		\item For every feasible solution~${\yStab}$ to~${\xStab}$, we can
		efficiently construct a feasible solution~${\ySSC}$ to~${\xSSC}$
		with
		\begin{equation*}
		\costSSC(\ySSC) - \coptSSC(\xSSC) \le \beta \cdot \left(\costStab(\yStab) - \coptStab(\xStab) \right)\formulaPunctuationSpace,
		\end{equation*}
		where~${\costSSC(\ySSC)}$ denotes the cardinality of~${\ySSC}$,
		and~${\costStab(\yStab)}$ denotes the total length of~${\yStab}$.
		\label{Lred:two}
	\end{listArabic}
	
	Given a \SPSC instance~${\xSSC}$ with~${5m}$ subsets we construct a
	\Stabbing instance~${\xStab}$ as in the proof of
	Theorem~\ref{thm:apx-card-stab} with the following
	specifications: Every rectangle corresponding to elements in~${A}$
	has width equal to~${1 + \delta}$ for~${\delta = 1 / 10m}$ and the
	intersection~${\intAs}$ of all these rectangles has width equal
	to~${1}$.  For~${1 \le t \le m}$, we choose the lengths of the line
	segments corresponding to~${\{w_t, x_t\}}$ and~${\{y_t, z_t\}}$ to
	be equal~${1}$.  Thus, every line segment has length either~${1}$
	or~${1+\delta}$.
	
	Consequently, any optimum solution to~${\xSSC}$ implies a feasible
	solution to~${\xStab}$ with cost at most~${(1+\delta)\cdot
		\coptSSC(\xSSC)}$.  Hence, we can bound~${\coptStab(\xStab)}$ from
	above by
	\[{(1+\delta)\cdot \coptSSC(\xSSC) < 2 \cdot
		\coptSSC(\xSSC)}\formulaPunctuationSpace.\]
	This shows Property~\ref{Lred:one} of~${L}$-reduction
	with~${\alpha=2}$.
	
	Now, given a feasible solution~${\yStab}$, we will first observe
	that it cannot consist of less line segments than an optimum
	solution.  Let~$x$ be the cardinality of any optimum solution
	to~$\xStab$. 
	Recall that every line segment has length~${1}$ or~${1 + \delta}$
	and that any feasible solution has at most~${5m}$ line segments.
	Thus, $x \le \coptStab(\xStab)$ and, consequently,
	\begin{alignat*}{2}
	x~&\le&~~& \costStab(\yStab) \\&\le&& |\yStab|(1+\delta) \\&\le&&
	|\yStab| + 5m\delta \\&=&& |\yStab| +
	\frac{1}{2}\formulaPunctuationSpace.
	\end{alignat*}
	Hence, the inequality holds only if~$\yStab$ contains at least~$x$
	lines segments.
	
	Next, let~${\ySSC}$ consist of all sets corresponding to the line
	segments in~${\yStab}$. Observe that~${\ySSC}$ is feasible.  To show
	Property~\ref{Lred:two}, we consider two cases.  In the first
	case,~${\yStab}$ has the same number of line segments as the optimum
	solution.  We immediately get
	\[{\costSSC(\ySSC) - \coptSSC(\xSSC) = 0}\] and the inequality of
	Property~\ref{Lred:two} holds.  In the second case,~${\yStab}$ has
	more line segments than the optimum solution.  Thus,~${x<|\yStab|
		\le 5m}$.  With the definition of~${\delta}$ we, obtain
	\begin{alignat*}{2}
	\costStab(\yStab) - \coptStab(\xStab) &>&~~& |\yStab| -
	x(1+\delta) \\&\ge&& 1-x\delta \\&\ge&& 1-5m\delta \\&=&&
	\frac{1}{2}\formulaPunctuationSpace.
	\end{alignat*}
	Furthermore, we have
	\begin{alignat*}{2}
	\coptStab(\xStab) &\le&~~& (1 + \delta) \cdot \coptSSC(\xSSC)
	\\&\le&& \coptSSC(\xSSC) + \delta 5m \\&\le&& \coptSSC(\xSSC) +
	\frac{1}{2}\formulaPunctuationSpace.
	\end{alignat*}
	On the other hand,~${\costSSC(\ySSC) \le \costStab(\yStab)}$.
	Putting things together, we get
	\begin{alignat*}{2}
	\costSSC(\ySSC) - \coptSSC(\xSSC) &\le&~~& \costStab(\yStab) -
	\coptStab(\xStab) + \frac{1}{2} \\&\le&& 2 \cdot
	\left(\costStab(\yStab) - \coptStab(\xStab)\right)
	\end{alignat*}
	and Property~\ref{Lred:two} holds for~${\beta = 2}$.
\end{proof}

\section{A Constant-Factor Approximation Algorithm for Stabbing}
\label{sec:const-appr}

In this section, we present a constant-factor approximation algorithm
for \Stabbing.  

First, we model \Stabbing as a set cover problem, and we revisit 
the standard linear programming relaxation for set cover and the
concept of shallow-cell complexity; see Sections~\ref{sec:lp}
and~\ref{sec:shall-cell-compl}. Then, we observe that there are \Stabbing
instances with high shallow-cell complexity.
This limiting fact prevents us from obtaining 
any constant approximation factor 
if applying the generalization of Chan et
al.~\cite{chan-etal12-weighted-geom-sc} in a direct way; see
Section~\ref{sec:shall-cell-compl}.
In order to bypass this limitation, we decompose any \Stabbing
instance into two disjoint families of low shallow-cell complexity.
Before describing the decomposition in
Section~\ref{sec:handl-gener-case}, we show how to merge solutions to
these two disjoint families in an approximation-factor preserving way;
see Section~\ref{sec:decomposition-lemma}.
Then, in Section~\ref{sec:x-laminar}, we observe that 
these families have sufficiently small shallow-cell complexity
to admit a constant-factor approximation.

\subsection{Set Cover and Linear Programming}
\label{sec:lp}

An instance~${(U,\mathcal{F}, c)}$ of weighted \problemName{Set Cover}
is given by a finite universe~${U}$ of~${n}$ elements, a family~${\cal
  F}$ of subsets of~${U}$ that covers~${U}$, and a cost
function~${c\colon \cal F\rightarrow\mathbb{Q}^+}$.  The objective is
to find a sub-family~$\cal S$ of~$\cal F$ that also covers~${U}$ and
minimizes the total cost~${c(\mathcal{S})}={\sum_{S \in \mathcal{S}}
  c(S)}$.

An instance~${(R,F)}$ of \ConstrainedStabbing, given by a set~$R$ of
rectangles and a set~$F$ of line segments, can be seen as a special
instance of weighted \problemName{Set Cover} where the rectangles
in~${R}$ are the universe~${U}$, the line segments in~${F}$ form the
sets in~${\mathcal{F}}$, and a line segment~${s\in F}$ ``covers'' a
rectangle~${r}$ if and only if~${s}$ stabs~${r}$.  Unconstrained
\Stabbing can be modeled by \problemName{Set Cover} as follows.  We
can, without loss of generality, consider only feasible solutions
where the end points of any line segment 
lie on the left or right boundaries of rectangles and
where each line segment touches the top boundary of some rectangle.
Thus, we can restrict ourselves to feasible solutions that are subsets
of a set~${F}$ of~${\bigOh(n^3)}$ candidate line segments.  This shows
that \Stabbing is a special case of \ConstrainedStabbing and, hence,
of \problemName{Set Cover}.

Let $(U,\mathcal{F},c)$ be an instance of \problemName{Set
  Cover}. The standard LP relaxation LP$(U,\mathcal{F},c)$ for this
instance is as follows.
\begin{alignat*}{2}
  &\text{Minimize} & \sum_{S\in \mathcal{F}} c(S)z_S\\
  &\text{subject to} \quad& \sum_{S\in\mathcal{F},S\ni e}z_S\geq 1 & \text{\quad for all } e\in U\\
  && z_S\geq 0 & \text{\quad for all } S\in \mathcal{F}\,.
\end{alignat*}
The optimum solution to this LP provides a lower bound on $\opt$. An
algorithm is called 
\emph{LP-relative $\alpha$-approximation algorithm} for a class~$\Pi$
of set cover instances if it rounds any feasible solution
$\mathbf{z}=(z_S)_{S \in \mathcal{F}}$
to the above standard LP relaxation for some instance
$(U,\mathcal{S},c)$ in this class to a feasible integral solution 
$\mathcal{S}\subseteq\mathcal{F}$ of cost~$c(\mathcal{S})\leq \alpha
\sum_{S\in \mathcal{F}} c(S)z_s$.

\subsection{Shallow-Cell Complexity}\label{sec:shall-cell-compl}
We define the shallow-cell complexity for classes that consist of
instances of weighted \problemName{Set Cover}.  Informally, the
shallow-cell complexity is a bound on the number of equivalent classes
of elements that are contained in a small number of sets.  Here is the
formal definition.
\begin{definition}[Chan et
  al.~\cite{chan-etal12-weighted-geom-sc}]\label{def:scc}
  Let~${f(m,k)}$ be a function non-decreasing in~${m}$ and~${k}$.  An
  instance~${(U,\mathcal{F},c)}$ of weighted \problemName{Set Cover}
  has shallow-cell complexity~${f}$ if the following holds for
  every~${k}$ and~${m}$ with~${1\le k \le m \le |\mathcal{F}|}$, and
  every sub-family~${\mathcal{S}\subseteq\mathcal{F}}$ of~${m}$ sets:
  All elements that are contained in at most~${k}$ sets
  of~${\mathcal{S}}$ form at most~${f(m,k)}$ equivalence classes
  (called \emph{cells}), where two elements are equivalent if they are
  contained in precisely the same sets of~${\mathcal{S}}$.  A class of
  instances of weighted \problemName{Set Cover} has shallow-cell
  complexity~${f}$ if all its instances have shallow-cell
  complexity~${f}$.
\end{definition}
Chan et al.\ proved that if a set cover problem has low shallow-cell
complexity then quasi-uniform sampling yields an LP-relative
approximation algorithm with good performance.
\begin{theorem}[Chan et
  al.~\cite{chan-etal12-weighted-geom-sc}]\label{thm:shall-cell-compl-apx}
  Let~${\phi(m)}$ be a non-decreasing function, and let~$\Pi$ be a
  class of instances of weighted \problemName{Set Cover}.  If~$\Pi$
  has shallow-cell complexity~${m\phi(m)k^{\bigOh(1)}}$, then~$\Pi$
  admits an LP-relative approximation algorithm (based on
  quasi-uniform sampling) with approximation
  ratio~${\bigOh(\max\{1,\log\phi(m)\})}$.
\end{theorem}

\begin{figure}[tb]
  \begin{subfigure}[b]{.47\textwidth}
    \centering
    \includegraphics[page=3]{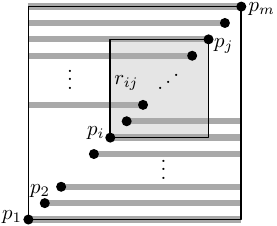}
    \caption{for \Stabbing}
    \label{fig:stabbing-counterexample}
  \end{subfigure}
  \hfill
  \begin{subfigure}[b]{.47\textwidth}
    \centering
    \includegraphics[page=4]{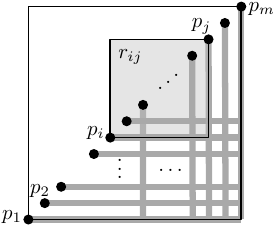}
    \caption{for \problemName{Horizontal--Vertical Stabbing}}
    \label{fig:hvs-counterexample}
  \end{subfigure}
  \caption{Instances with high shallow-cell complexity.}
  \label{fig:scc-counterexample}
\end{figure}
Unfortunately, there are instances of \Stabbing that have high
shellow-cell complexity, so we cannot directly obtain a sub-logarithmic
performance via Theorem~\ref{thm:shall-cell-compl-apx}.  These
instances can be constructed as follows; see
\fig~\ref{fig:stabbing-counterexample}.  Let~${m}$ be an even positive
integer.  For~${i=1,\dots,m}$, define the point~${p_i=(i,i)}$.  For
each pair~${i,j}$ with~${1\le i \le m/2 < j \le m}$, let~${r_{ij}}$ be
the rectangle with corners~${p_i}$ and~${p_j}$.  Now, consider the
following set~${\mathcal{S}}$ of~${m}$ line segments.
For~${i=1,\dots,m/2}$, the set~${\mathcal{S}}$ contains the
segment~${s_i}$ with endpoints~${p_i}$ and~${(m,i)}$.
For~${i=m/2+1,\dots,m}$, the set~${\mathcal{S}}$ contains the
segment~${s_i}$ with endpoints~${(1,i)}$ and~${p_i}$.

We want to count the number of rectangles that are stabbed by at most
two segments in~${\mathcal S}$.  Consider any~${i}$ and~${j}$
satisfying~${1\leq i\leq m/2<j\leq m}$.  Observe that the
rectangle~${r_{ij}}$ is stabbed precisely by the segments~${s_i}$
and~${s_j}$ in~${\mathcal{S}}$.  Hence, according to
Definition~\ref{def:scc}, our instance consists of at least~${m^2/4}$
equivalence classes for~${k=2}$.  Thus, if our instance has shallow
cell-complexity $f$ for some suitable function $f$, we
have~${f(m,2)\finO \Omega(m^2)}$.  Since~${f}$ is non-decreasing, we
also have~${f(m,k)\finO \Omega(m^2)}$ for~${k\ge 2}$.  Hence,
Theorem~\ref{thm:shall-cell-compl-apx} implies only
an~${\bigOh(\log n)}$-approximation algorithm for \Stabbing where we
use the above-mentioned fact (see Section~\ref{sec:lp}) that we can
restrict ourselves to $m=O(n^3)$ many candidate segments.

\subsection{Decomposition Lemma for Set
  Cover}\label{sec:decomposition-lemma}

Our trick is to decompose general instances of \Stabbing (which may
have high shallow-cell complexity) into partial instances of low
complexity with a special, laminar structure.  We use the following
simple decomposition lemma, which holds for arbitrary set cover
instances.
\begin{lemma}\label{lem:decomp-set-cover}	
  Let~$\Pi$,~$\Pi_1$,~$\Pi_2$ be classes of
  \problemName{Set Cover} where~$\Pi_1$ and~$\Pi_2$
  admit LP-relative~$\alpha_1$- and~$\alpha_2$-approximation
  algorithms, respectively.  The class~$\Pi$ admits an
  LP-relative~${(\alpha_1+\alpha_2)}$-approximation algorithm if, for
  every instance~$(U,\mathcal{F},c)\in \Pi$, the 
  family~$\mathcal{F}$ can be partitioned into~$\mathcal{F}_1,\mathcal{F}_2$
  such that, for any partition of~$U$ into~$U_1,U_2$ where~$U_1$ is
  covered by~$\mathcal{F}_1$ and~$U_2$ by~$\mathcal{F}_2$, the
  instances~$(U_1,\mathcal{F}_1,c)$ and $(U_2,\mathcal{F}_2,c)$ are
  instances of~$\Pi_1$ and~$\Pi_2$, respectively.
\end{lemma}
\begin{proof}	
  Let $\mathbf{z}=(z_S)_{S \in \mathcal{F}}$ be a feasible solution to
  LP$(U,\mathcal{F},c)$. Let $U_1,U_2=\emptyset$ initially. Consider
  an element $e\in U$. Because of the constraint
  $\sum_{S\in\mathcal{F},S\ni e}z_S\geq 1$ in the LP relaxation and
  because of~$\mathcal{F}=\mathcal{F}_1\cup\mathcal{F}_2$, 
  at least one of the two cases~${\sum_{S\in\mathcal{F}_1,S\ni
    e}z_S\geq \alpha_1/(\alpha_1+\alpha_2)}$ 
    and~$\sum_{S\in\mathcal{F}_2,S\ni e}z_S\geq
  \alpha_2/(\alpha_1+\alpha_2)$ occurs.  In the first case, we add $e$
  to $U_1$.  In the second case, we add $e$ to $U_2$.  We execute this
  step for each element $e\in U$.

  Now, consider the instance $(U_1,\mathcal{F}_1,c)$. For each
  $S\in\mathcal{F}_1$, set
  $z_S^1:=\min\{(\alpha_1+\alpha_2)/\alpha_1z_S,1\}$.  Since
  $\sum_{S\in\mathcal{F}_1,S\ni e}z_S\geq
  \alpha_1/(\alpha_1+\alpha_2)$ for all $e\in U_1$, we have that
  $\mathbf{z}^1=(z_S^1)_{S \in \mathcal{F}_1}$ forms a feasible solution
  to LP$(U_1,\mathcal{F}_1,c)$. Next, we apply the LP-relative
  $\alpha_1$-approximation algorithm to this instance to obtain a
  solution $\mathcal{S}_1\subseteq\mathcal{F}_1$ that covers
  $U_1$ and whose cost is at 
  most~${\alpha_1 \sum_{S\in \mathcal{F}_1} c(S)z_S^1\leq
  (\alpha_1+\alpha_2)\sum_{S\in \mathcal{F}_1} c(S)z_S}$. Analogously, we can compute a 
  solution~$\mathcal{S}_2\subseteq\mathcal{F}_2$ to $(U_2,\mathcal{F}_2,c)$ of
  cost at most~$(\alpha_1+\alpha_2)\sum_{S\in \mathcal{F}_2}
  c(S)z_S$. 

  To complete the proof, note that
  $\mathcal{S}_1\cup\mathcal{S}_2$ is a feasible solution to
  $(U,\mathcal{F},c)$ of cost at most~$(\alpha_1+\alpha_2)\sum_{S\in
    \mathcal{F}_1\cup\mathcal{F}_2} c(S)z_S$.  Hence, our algorithm 
  is an LP-relative $(\alpha_1+\alpha_2)$-approximation algorithm.
\end{proof}

\subsection{x-Laminar Instances} 
\label{sec:x-laminar}

\begin{definition}\label{def:x-laminar-instances}
  An instance of \ConstrainedStabbing is called \emph{$x$-laminar} if
  the projection of the segments in this instance onto the $x$-axis
  forms a laminar family of intervals. That is, any two of these
  intervals are either interior-disjoint or one is contained in the
  other.
\end{definition}

\begin{lemma}\label{lem:scc-x-laminar-instances}
  The shallow-cell complexity of an $x$-laminar instance of
  \ConstrainedStabbing can be upper bounded by $f(m,k)=mk^2$.  Hence,
  such instances admit a constant-factor LP-relative approximation
  algorithm.
\end{lemma}
\begin{proof}
  To prove the bound on the shallow-cell complexity, consider a set
  $\mathcal{S}$ of $m$ segments.  Let $1\leq k\leq m$ be an
  integer. Consider an arbitrary rectangle $r$ that is stabbed by at
  most $k$ segments in $\mathcal{S}$. Let $\mathcal{S}_r$ be the set
  of these segments. Consider a shortest segment~$s\in\mathcal{S}_r$.
  By laminarity, the projection of any segment in $\mathcal{S}_r$ onto
  the $x$-axis contains the projection of $s$ onto the $x$-axis. Let
  $C_s=(s_1,\dots,s_\ell)$ be the sequence of \emph{all} segments in
  $\mathcal{S}$ whose projection contains the projection of~$s$,
  ordered from top to bottom. The crucial point is that the set~$\mathcal{S}_r$ 
  forms a contiguous sub-sequence
  $s_i,\dots,s_{i+|\mathcal{S}_r|-1}$ of $C_s$ that contains $s=s_j$
  for some~$i\leq j\leq i+|\mathcal{S}_r|-1$.  Hence, $\mathcal{S}_r$
  is uniquely determined by the choice of $s\in\mathcal{S}$ (for which
  there are $m$ possibilities), the choice of $s_i$ with
  $i\in\{j-k,\dots,j\}$ within the sequence~$C_s$ (for which there are
  at most $k$ possibilities), and the cardinality of~$\mathcal{S}_r$
  (for which there are at most $k$ possibilities). This implies that
  $\mathcal{S}_r$ is one of $mk^2$ many sets that define a cell. This
  completes our proof since $r$ was picked arbitrarily.
\end{proof}

\subsection{Decomposing General Instances into Laminar Instances}
\label{sec:handl-gener-case}

\begin{lemma}\label{lem:decomp-into-laminar}
  Given an instance $I$ of (unconstrained) \Stabbing with rectangle
  set $R$, we can compute an instance $I'=(R,F)$ of
  \ConstrainedStabbing with the following properties.  The set $F$ of
  segments in $I'$ has cardinality $O(n^3)$, it can be decomposed into
  two disjoint $x$-laminar sets $F_1$ and $F_2$, and ${\opt_{I'}\leq
    6\cdot\opt_{I}}$.
\end{lemma}
\begin{proof}
Let~$F'$ be the set of~$O(n^3)$ candidate segments as defined in Sec.~\ref{sec:lp}: 
For every segment~$s$ of~$F'$, the left endpoint of~$s$ 
  lies on the left boundary of some rectangle, the right endpoint
  of~$s$ lies on the right boundary of some rectangle, and~$s$
  contains the top boundary of some rectangle.
Recall that~$F'$ contains the optimum solution.

  Below, we stretch each of the segments in $F'$ by a factor of at
  most $6$ to arrive at a set~$F$ of segments having the claimed
  properties. By scaling the instance we may assume that the longest
  segment in $F'$ has length $1/3$.

  For any $i,j\in\mathbb{Z}$ with $i\geq 0$, let $I_{ij}$ be the
  interval $[j/2^i,(j+1)/2^i]$. Let $\mathcal{I}_1$ be the family of all
  such intervals $I_{ij}$. We say that $I_{ij}$ has level $i$. Note
  that $\mathcal{I}_1$ is an $x$-laminar family of intervals
  (segments). Let $\mathcal{I}_2$ be the family of intervals that
  arises if each interval in $\mathcal{I}_1$ is shifted to the right by
  the amount of~$1/3$.  That is, $\mathcal{I}_2$ is the family of
  all intervals of the form~$I_{ij}+1/3:=[j/2^i+1/3,(j+1)/2^i+1/3]$ (for any $i,j\in\mathbb{Z}$ with $i\geq 0$). Clearly,
  $\mathcal{I}_2$ is~$x$-laminar, too.
  
  We claim that any arbitrary interval $J=[a,b]$ of length at most
  $1/3$ is contained in an interval $I$ that is at most $6$ times
  longer and that is contained in $\mathcal{I}_1$ or in
  $\mathcal{I}_2$. This completes the proof of the lemma since then
  any segment in~$F'$ can be stretched by a factor of at most~$6$ so
  that its projection on the $x$-axis lies in $\mathcal{I}_1$ (giving
  rise to the segment set $F_1$) or in $\mathcal{I}_2$ (giving rise to
  the segment set $F_2$). Setting $F=F_1\cup F_2$ completes the
  construction of the instance $I'=(R,F)$.

  To show the above claim, let $s$ be the largest non-negative integer
  with~${b-a\leq 1/(3\cdot 2^{s})}$. If $J$ is contained in the
  interval $I_{sj}$ for some integer $j$, we are done 
  because~${b-a>1/(6\cdot 2^{s})}$ by the choice of $s$. If $J$ is not contained
  in any interval $I_{sj}$, then there exists some integer~$j$ such
  that $j/2^{s}\in J=[a,b]$ and thus $a\in I_{s,j-1}$. Since
  $b-a\leq 1/(3\cdot 2^{s})$, we have that $J$ is completely contained
  in the interval $I':=I_{s,j-1}+1/(3\cdot 2^{s})$ and in the 
  interval~${I'':=I_{s,j}-1/(3\cdot 2^{s})}$.
  
  We complete the proof by showing that one of the intervals $I', I''$
  is actually contained in~$\mathcal{I}_2$. To this end, note that
  $1/3=\sum_{\ell=1}^\infty(-1)^{\ell-1}/2^\ell$.
  Hence, if $s$ is even, the interval $I'-1/3$ lies in $\mathcal{I}_1$,
  and if $s$ is odd, the interval $I''-1/3$ lies in $\mathcal{I}_1$.
\end{proof}
Now, we apply the decomposition lemma to
Lemmas~\ref{lem:scc-x-laminar-instances}
and~\ref{lem:decomp-into-laminar} and obtain our main result.
\begin{theorem}\label{thm:const-fact-stab}
  \Stabbing admits a constant-factor LP-relative approximation algorithm.
\end{theorem}

Complementing Lemmas~\ref{lem:scc-x-laminar-instances}
and~\ref{lem:decomp-into-laminar}, 
\fig~\ref{fig:stabbing-counterexample} shows that the union of two
$x$-laminar families of segments may have shallow-cell complexity with
quadratic dependence on~$m$.  This
shows that the property of having low shallow-cell complexity
is not closed under taking unions.

\section{Further Applications of the Decomposition Lemma}
\label{sec:applications}

In this section we demonstrate that our decomposition technique can be
applied in other settings, too.

\paragraph{Horizontal--Vertical Stabbing.}
In this new variant of \Stabbing, a rectangle may
be stabbed by a horizontal or by a vertical line segment (or by both).
Using the results of Section~\ref{sec:handl-gener-case} and the
decomposition lemma where we decompose into horizontal and vertical
segments, we immediately obtain the following result.
\begin{corollary}\label{cor:hor-vert-stabbing}
  \problemName{Horizontal--Vertical Stabbing} admits an LP-relative
  constant-factor approximation algorithm.
\end{corollary}
\Fig~\ref{fig:hvs-counterexample} shows that a laminar family of
horizontal segments and vertical segments may have a shallow-cell
complexity with quadratic dependence on $m$.  Thus,
Corollary~\ref{cor:hor-vert-stabbing} is another natural example where
low shallow-cell complexity is not closed under union and where the
decomposition lemma 
gives a constant-factor approximation although the shallow-cell
complexity is high.

\paragraph{Stabbing 3D-Boxes by Squares.}
In the 3D-variant of \Stabbing, we want to stab 3D-boxes with
axis-aligned squares, minimizing the sum of the areas or the sum of
the perimeters of the squares.  Here, ``stabbing'' means ``completely
cutting across''.
By combining the same idea with shifted quadtrees---the 2D-equivalent
of laminar families of intervals---we obtain a constant-factor
approximation for this problem.  It is an interesting question if our
approach can be extended to handle also arbitrary rectangles but this
seems to require further ideas.

\paragraph{Covering Points by Anchored Squares.}
Given a set~$P$ of points that need to be covered and a set~$A$ of
anchor points, we want to find a set of axis-aligned squares such
that each square contains at least one anchor point,
the union of the squares covers~$P$, and the total area or the total perimeter 
of the squares is minimized.  Again, with the help of shifted
quadtrees, we can apply the decomposition lemma.  In this case, we
do not even need to apply the machinery of quasi-uniform sampling;
instead, we can use dynamic programming on the decomposed instances.
This yields a deterministic algorithm with a concrete constant
approximation ratio ($4 \cdot 6^2$, without polishing).

\section{Conclusion}

We have seen that \Stabbing is \NP-hard and that it admits
an~${\bigOh(1)}$-approxi\-mation algorithm.  Since our positive
results relies on a general result regarding the shallow-cell
complexity of the problem, it would be interesting to design a direct,
combinatorial $c$-approxi\-mation algorithm with a concrete
constant~$c$ that makes use of the geometry underlying the problem.

On the negative side, it remains open whether \Stabbing is \APX-hard,
which is the case for \ConstrainedStabbing and \CardinalityStabbing.
Do the latter two problems admit
constant-factor approximation algorithms?  So far, we have only
an~${\bigOh(\log\log\opt)}$-approximation algorithm for
\CardinalityStabbing via an existing approximation algorithm for
piercing 3D-boxes \cite{aronov10-eps-nets-rectangles},
see Corollary~\ref{cor:loglogopt} in Appendix~\ref{sec:piercing}.
(Here, $\opt$ denotes the size of an optimum solution.)

Finally, it would be interesting to examine natural problems of high
shallow-cell complexity of unsettled approximability and try to
partition them (possibly by our decomposition technique) into
instances of low-shallow cell complexity, as in
Section~\ref{sec:applications}.

\appendix

\section{Structural Properties and Applicability of Existing Techniques}
\label{sec:relat-probs}

Since our problem is---at least in its general form in the setting of
line segments---new, we investigate the applicability of existing
techniques for (geometric) \problemName{Set Cover}.  We provide
instances of \Stabbing where the greedy algorithm (and natural
variants of it) have performance~${\Omega(\log n)}$; see
Section~\ref{sec:greedy}.  Then we explore the structural relation of
\Stabbing to existing geometric set cover (or, equivalently, hitting
set) problems; see Section~\ref{sec:piercing}.

\subsection{Greedy Algorithm for Set Cover}
\label{sec:greedy}

The greedy algorithm has approximation ratio~${\ln n}$ for
\problemName{Set Cover} on~${n}$ elements.  It is known that this
result is the best possible unless
$\runtimeclass{P}=\NP$~\cite{feige98-setcover}.

The greedy algorithm---translated to \Stabbing---works as follows.
Start with an empty set~${S}$ of segments.  Pick a segment~${s}$ that
minimizes the cost efficiency~${\len{s}/n_s}$ where~${n_s}$ is the
number of rectangles that are stabbed by~${s}$.  Add~${s}$ to~${S}$
and remove the rectangles that are stabbed by~${s}$ from~${R}$.
Repeat these steps until~${R}$ becomes empty.  Eventually, output the
resulting set~${S}$.  This algorithm certainly has approximation
ratio~${\bigOh(\log n)}$ for \Stabbing.

While the bound~${\bigOh(\log n)}$ is tight for general
\problemName{Set Cover} this does not immediately imply tightness for
\Stabbing as well.  Unfortunately, there are instances of \Stabbing
where the greedy algorithm (and natural variants of it) have
ratio~${\Omega(\log n)}$.

\begin{figure}[tb]
	\centering
	\includegraphics{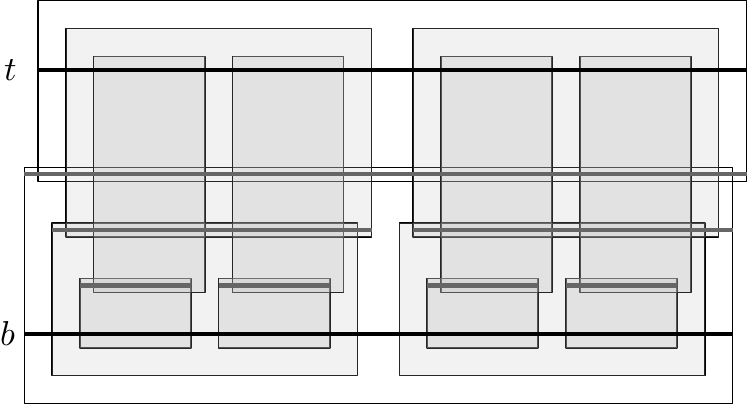}
	\caption{Instance where the greedy algorithm has
		performance~${\Omega(\log n)}$.  The black segments belong to the
		optimum solution and the gray segments belong to the output of the
		greedy algorithm.  To make the drawing easier to read, we moved
		the rectangles of~${T}$ (those stabbed by~${t}$) slightly to the
		right and to the bottom. In our instance, the bottom edges of the
		rectangles in~${T}$ coincide with the top edges of their
		counterparts in~${B}$ (which are stabbed by~${b}$), and there are
		no two top edges with the same vertical projection.  }
	\label{fig:greedy-counterexample}
\end{figure}
Consider the instance shown in \fig~\ref{fig:greedy-counterexample}.
We introduce two segments~${t}$ and~${b}$ of length~${1}$.  Then we
construct a set~${B}$ of nested rectangles that are all stabbed
by~${b}$.  The set~${B}$ is subdivided into levels~${0, 1,\dots,
	\ell}$ according to the nesting hierarchy (see figure).  At
level~${i}$, there are~${2^i}$ pairwise disjoint rectangles of
width~${(1-i\epsilon)/2^i}$ for a sufficiently small
positive~${\epsilon}$.  We slightly perturb the top edges of the
rectangles in~${B}$ so that the top edges of the rectangles in~${B}$
have pairwise different~$y$-coordinates.  Next, we construct a
set~${T}$ of rectangles.  For each rectangle~${r\in B}$ we create a
corresponding rectangle~${r'}$ in~${T}$ of the same width such that
the bottom edge of~${r'}$ coincides with the top edge of~${r}$
and~${r'}$ is stabbed by~${t}$.

We now analyze how the greedy algorithm performs on this instance.
First, we verify that the first segment~$s$ picked by the algorithm
contains the top edge of some rectangle~${r\in B}$ and has endpoints
that lie on the left and right edge of some rectangle~${r'\in B}$.
If~${s}$ were not containing the top edge of any rectangle in~${B}$,
we could vertically move it until it contains such a top edge and
simultaneously stabs one rectangle more than before; a contradiction
to the greedy choice.  On the other hand, if the endpoints of~${s}$
were not lying on a left and right edge of the same rectangle, then,
by our construction, there would be a small positive interval on~${s}$
which is not contained in the rectangles stabbed by~${s}$. We could
cut the interval out of~${s}$ and obtain one or two new line segments,
where at least one of them has a better cost efficiency than~${s}$; a
contradiction.

Now, consider a segment~${s}$ that is lying on the top edge of some
rectangle~${r\in B}$ and is containing the vertical boundaries of some
rectangle~${r'\in B}$.  Let~${i\in\{0,\dots,\ell\}}$ be the level
of~${r'}$.  Observe that~${s}$ has length~${(1-i\epsilon)/2^i}$ and
stabs~${\sum_{j=0}^{\ell-i} 2^j + 1 = 2^{\ell-i+1}}$ many rectangles.
(Note that~${s}$ also stabs the rectangle corresponding to~${r}$
in~${T}$.)  Therefore,~${s}$ has cost
efficiency~${(1-i\epsilon)/2^{\ell+1}}$, which is minimized for the
biggest-possible value of~${i}$.  From this we can conclude that the
algorithms picks~${s}$ such that~${r'}$ belongs to the highest
level~${i}$, which implies that~${r'}$ and~${r}$ coincide. Thus,~${s}$
is the top edge of some rectangle in~${B}$ with the highest
level~${i}$.  In subsequent iterations, the algorithm continues
selecting the top edge of a rectangle in~${B}$ that has the highest
level among the remaining rectangles.  Overall the algorithm produces
a solution that consists of all top edges of rectangles in~${B}$ which
has cost~${\Omega(\log n)}$ since the highest level~${\ell}$ is
in~${\Omega(\log n)}$.  The solution~${\{t,b\}}$, however, has only
cost~${2}$, which completes our claim.

The example above suggests the following natural variation of the
greedy algorithm.  In each step, pick the segment that minimizes the
ratio of its length to the total \emph{width} of the (previously
unstabbed) rectangles it stabs.  We can easily modify the instance so
that also this algorithms performs bad.  In the first step, we remove
all rectangles of odd levels and do not change the level
enumeration. 
Thus, all levels are now even.  In the second step, we create copies
of each rectangle so that a rectangle at level~${i}$ has
multiplicity~${\ceil{ 2^i/(1-i\epsilon)}}$.  This multiplicity will
ensure that the total weight of equivalent rectangles is roughly~${1}$
and not smaller than~$1$.  Note that the number of levels is still
in~${\Omega(\log n)}$ although we increased the number~$n$ of
rectangles.

To this end, we show that the modified greedy algorithm picks again
all top edges of the rectangles in~${B}$, always greedily picking one
from the currently highest level.  Suppose this were not the case and
consider the first segment~${s}$ not picked in this manner.  By the
same discussion as in the unweighted case, the segment~${s}$ lies on a
top edge of a rectangle~${r\in B}$ and is touching the horizontal
boundaries of a rectangle~${r'\in B}$.  Let~${i}$ be the level
of~${r'}$. By our assumption,~${i}$ is not the highest level, hence,
the highest level is at least~${i+2}$ (as all levels are even).  Thus,
we can find a segment~${s'}$ that lies on a top edge of some
rectangle~${r''\in B}$ and, at the same time, touches the horizontal
boundaries of some rectangle of level~${i+2}$.

Let~${w}$ be the total width of the rectangles stabbed by~${s'}$
excluding the rectangles of~${T}$ corresponding to~${r''}$.  Note that
the total width of those excluded rectangles is at least~${1}$.
Hence, the cost efficiency of~${s'}$ is at
most \[\frac{1-(i+2)\epsilon}{2^{i+2} (w +
	1)}\formulaPunctuationSpace.\]

Next, consider~${s}$. The total width of~${r'}$ and its copies is at
most
\[\ceil{\frac{2^i}{1-i\epsilon}}\cdot\frac{1-i\epsilon}{2^i} < 1 +
\frac{1-i\epsilon}{2^i}\formulaPunctuationSpace.\] The same bound
holds for the total width of the rectangles of~${T}$ corresponding
to~${r}$ as the level of~${r}$ is not smaller than~${i}$.  Thus, the
total width of the rectangles stabbed by~${s}$ is at most
\[4w + 2\left(1 + \frac{1-i\epsilon}{2^i}\right) < 4w +
4\formulaPunctuationSpace.\] Hence, the cost efficiency of~${s'}$ is
greater than \[\frac{1-i\epsilon}{2^i \left(4w + 4\right)} =
\frac{1-i\epsilon}{2^{i+2} \left(w + 1\right)}\] and thus bigger than
the cost efficiency of~${s'}$; a contradiction to the greedy choice.

Note that none of the segments returned by the algorithm is redundant
so that a post-processing that removes unnecessary segment parts does
not help.

\subsection{Relation to \problemName{Piercing}}
\label{sec:piercing}

In this section, we consider how our stabbing problems relate to the
well-studied hitting set problem for axis-aligned rectangles (or boxes
in higher dimensions), which we call \problemName{Piercing}. In this
problem, we are given a set~${R}$ of axis-aligned rectangles (or
boxes) and a set~${P}$ of points.  We want to hit all rectangles using
a minimum number of points from~${P}$.  We also consider the weighted
version where each point has a positive weight and we want to minimize
the total weight of the points selected.  Similarly to \Stabbing, also
this problem can be expressed naturally in terms of \problemName{Set
	Cover}: The rectangles are the elements to be covered, and the
piercing points are the sets.  This correspondence allows us to
compare stabbing and piercing by asking whether a given set cover
instance has a realization as either of them.  We will show that every
stabbing instance corresponds directly in this way to a piercing
instance in dimension three.  Just in dimension two, however, not
every stabbing instance can be realized as a piercing instance.  This
shows that \Stabbing is structurally different from two-dimensional
\problemName{Piercing}.

\begin{theorem}
	Any set cover instance~${(\universe,\sets)}$ arising from \Stabbing
	can be realized as an instance of weighted \problemName{Piercing} in
	dimension~${3}$.
\end{theorem}

\begin{proof}
	Starting with a ($2$-dimensional) stabbing instance, we will
	translate it to a~${3}$-dimensional piercing instance: Every
	rectangle becomes an axis-aligned box and every stabbing line
	segment becomes a piercing point.  Note that a stabbing line segment
	is defined by an interval~${[\Ll,\Lr]}$ and a height \Ly{}.  We lift
	it to the~${3}$-dimensional point~${(\Ll,\Lr,\Ly)}$ and assign it
	the weight~${\Lr-\Ll}$.  Consider a
	rectangle~${[\Rl,\Rr]\times[\Rb,\Rt]}$.  The line segment stabs this
	rectangle if and only if
	\begin{equation}
	\Ll\leq\Rl, \quad \Rr\leq\Lr,\quad \mathrm{and}\quad
	\Ly\in[\Rb,\Rt]\formulaPunctuationSpace. 
	\end{equation}
	This describes an axis-aligned box that is unbounded on one side
	of~$\Ll$ on the first coordinate axis and on one side of~$\Lr$ on
	the second coordinate axis. 
	We can observe that an optimal solution does not need to use any
	line segments with endpoints to the left of all rectangles or to the
	right of all rectangles.  This limits the relevant values of \Ll{}
	and \Lr{} and we can bound the box on all sides.
\end{proof}

Aronov et al.~\cite{aronov10-eps-nets-rectangles} describe
an~${\bigOh(\log\log\opt)}$-approximation algorithm for unweighted
\problemName{Piercing} in dimension~${3}$, where~${\opt}$ is the size
of a an optimum solution. This algorithm immediately gives us the same
bound for \CardinalityStabbing.  Their result does not carry over to
weighted \problemName{Piercing}, so we cannot use it to solve
\Stabbing.

\begin{corollary}
	\label{cor:loglogopt}
	There is an~${\bigOh(\log\log\opt)}$-approximation algorithm for
	\CardinalityStabbing, where~${\opt}$ is the size of an optimum
	solution.
\end{corollary}

Now, we show that such a correspondence does not exist in
dimension~${2}$: There exist stabbing instances that have no
corresponding piercing instance.  A set~${S\in\sets}$ in a set cover
instance~${(\universe,\sets)}$ is called \emph{universal}
if~${S=\universe{}}$.  Note that the universal set (if there exists
any) is not necessarily an optimum solution since we are dealing with
weighted \problemName{Set Cover}.

\begin{lemma} \label{lem:smallPiercing} Let~${(\universe,\sets)}$ be a
	\problemName{Set Cover} instance on~${n}$ elements that arises from
	a \problemName{Piercing} instance and contains the universal set.
	For any~${k}$,~${\sets}$ contains~${\bigOh(n)}$ distinct sets of
	cardinality~${k}$.
\end{lemma}
\begin{proof}
	Consider the faces of the arrangement on the plane induced by the
	set~${R}$ of~${n}$ rectangles of the \problemName{Piercing}
	instance.  Any points in the same face pierce exactly the same set
	of rectangles and are therefore the same set in terms of \sets{}.
	Call the number of rectangles pierced by points in a face the
	\emph{depth} of the face.  Since it is given that~${C}$ contains the
	universal set, there must be a face of depth~${n}$.  This face
	contains a point~${\unipoint\in P}$ and this point pierces all
	rectangles.  See \fig~\ref{fig:unipierce} for an example, where
	\theTick{} indicates \unipoint{}.
	
	\begin{figure}[tb]
		\centering 
		\subcaptionbox{\label{fig:unipierce}A \problemName{Piercing}
			instance with a universal
			point.}{~~\includegraphics{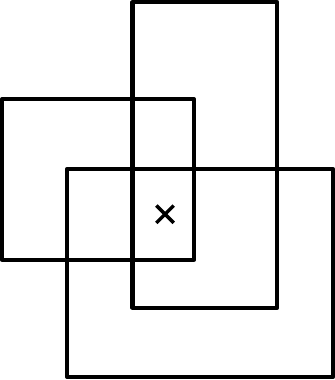}~~} \hfill
		\subcaptionbox{\label{fig:slabs}Slabs with depth
			values.}{\includegraphics{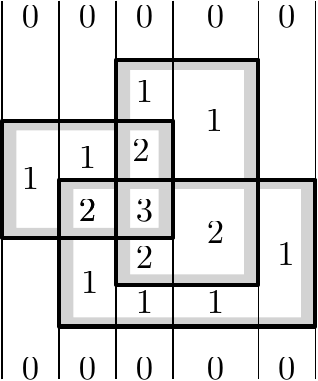}} \hfill
		\subcaptionbox{\label{fig:largeStabbing}A \Stabbing instance with
			many lines of equal
			cardinality.}{~~\includegraphics[page=4]{largeStabbing}~~}
		\caption{Some structural properties of \problemName{Piercing} and
			\Stabbing.}
	\end{figure}
	
	Now, we consider a vertical line at every left and right edge of a
	rectangle.  This cuts the plane into~${\bigOh(n)}$ vertical slabs
	and within each slab, all faces are rectangles (see
	\fig~\ref{fig:slabs}).  In each slab, the topmost face has depth
	zero.  Traversing downward until the height of \unipoint{}, every
	next face \emph{increases} the depth by at least one.  Traversing
	further downward \emph{decreases} the depth by at least one for each
	face.  Hence, for any~${k}$, there are at most two faces with
	depth~${k}$ in a slab.  The number of faces of a certain depth
	bounds the number of distinct sets of that size, and the claimed
	bound follows.
\end{proof}

\begin{lemma} \label{lem:largeStabbing} For every odd~${n}$, there
	exist \problemName{Set Cover} instances on~${n}$ elements arising
	from \Stabbing instances that contain the universal set
	and~${\Omega(n^2)}$ distinct sets of equal cardinality.
\end{lemma}

\begin{proof}
	Let~${\ell}$ be arbitrarily large and even.  For
	each~${i\in\{-\ell,..,\ell\}}$, we introduce a rectangle~${r_i}$.
	Thus, we have~${n=2\ell+1}$ rectangles, and we will place them in a
	double staircase as follows; see \fig~\ref{fig:largeStabbing}.  All
	rectangles have width~${1}$ and touch the~$x$-axis with their bottom
	edges.  For each~${i\in\{-\ell,..,\ell\}}$, the left edge of
	rectangle~${r_i}$ has $x$-coordinate~${i}$ and height~${|i|+1}$.
	Call the rectangles with negative index \emph{left} and the ones
	with positive index \emph{right}.  A stabbing line is said to have
	\emph{level} $i$ if its~${y}$-coordinate is in~${(i,i+1)}$.  At
	level~${0}$, we add a stabbing line that stabs every rectangle.
	
	Let~${k=\ell/2}$.  Now, we construct~${k}$ stabbing lines on each of
	the levels~${1}$ through~${k+1}$, each stabbing~${k+1}$ rectangles.
	Consider level~${i}$.  For every~${j}$ with~${1\leq j\leq k}$, the
	stabbing line~${s_{i,j}}$ stabs~${j}$ many left rectangles
	and~${k+1-j}$ many right rectangles.  This construction is uniquely
	defined, enough rectangles exist on these levels, and all of these
	line segments stab distinct sets of~$k+1$ rectangles.  Thus, the
	lemma holds: by construction, we have a universal set,
	and~${k\cdot(k+1)\finO \Omega(n^2)}$.
\end{proof}

\begin{theorem}
	There exist \problemName{Set Cover} instances that are realizable
	as~${2}$-dimensional \Stabbing but not as~${2}$-dimensional
	\problemName{Piercing}.
\end{theorem}

\begin{proof}
	Consider an arbitrarily large \problemName{Set Cover} instance
	on~${n}$ elements from Lemma~\ref{lem:largeStabbing}.  It is
	realizable as a \Stabbing instance, has the universal set, and
	contains~${\Omega(n^2)}$ distinct sets of equal size.  If it is
	large enough, then, by Lemma~\ref{lem:smallPiercing}, it does not
	have a realization as \problemName{Piercing}.
\end{proof}

\end{document}